\documentclass[a4paper,UKenglish]{article}

\usepackage{microtype}
\usepackage{fullpage}
\usepackage{graphics}
\graphicspath{{./figures/},{./figures/}}

\usepackage{authblk}

\usepackage{graphicx}
\usepackage{amsthm}
\usepackage{amsmath}
\usepackage{enumerate}
\usepackage{color}
\usepackage{xspace}
\usepackage{url}
\usepackage{microtype}
\usepackage{graphicx}
\usepackage{subfigure}
\usepackage{booktabs} 
\usepackage[table,xcdraw]{xcolor}
\usepackage{hyperref}
\usepackage{verbatim}

\usepackage{graphicx}
\usepackage{amsmath}
\usepackage{amssymb}
\usepackage{xcolor}
\usepackage{cite}
\usepackage[textsize=tiny]{todonotes}
\graphicspath{ {../../images/} }
\usepackage{glossaries}
\usepackage{mathtools} 
\newglossary{symbols}{sym}{sbl}{Maximum Clique in Two-sided L-shape Graph}
\newglossaryentry{fn}{type=symbols,name={$D(a,b)$},description={is the maximum clique on a $Two-sided$ $L-shape$ graph which includes $a,b$ and all other $L$s that are out side of the $a-b$ interval and have height smaller than minimum height of $a$ and $b$.}}
\newglossaryentry{fn2}{type=symbols,name={$D(c,d)$},description={d is the furthest $L-shape$ from $a$ and $b$}}

\newcommand\testletter[1]{   \raisebox{\depth}{\rotatebox{180}{\textsf{#1}}} 
}

\usepackage{amsfonts}\usepackage{amssymb}
\usepackage{thmtools}\usepackage{mathscinet}
\usepackage{thm-restate}
    \newtheorem{theorem}{Theorem}
    
    \newtheorem{lemma}[theorem]{Lemma}

    	\definecolor{darkgreen}{rgb}{0.01, 0.93, 0.29}
\definecolor{lightbrown}{rgb}{0.91, 0.4, 0.11}
\usepackage{framed}

\title{Finding a Maximum Clique in a Grounded 1-Bend String Graph}

\author[1]{J. Mark Keil}
\author[1]{Debajyoti Mondal} 
\author[1]{Ehsan Moradi} 
\author[2]{Yakov Nekrich} 
 
\affil[1]{Department of Computer Science, University of Saskatchewan, Saskatoon, Canada\\
  \texttt{mark.keil@cs.usask.ca}, \texttt{dmondal@cs.usask.ca}, \texttt{e.moradi@usask.ca}}  
\affil[2]{Department of Computer Science, Michigan Technological Univeristy, Michigan, USA\\
  \texttt{yakov@mtu.edu}}

\usepackage[textsize=tiny]{todonotes}
\usepackage{verbatim}

\begin{document}
\maketitle              
\begin{abstract}
A grounded 1-bend string graph is an intersection graph of a set of polygonal lines, each with one bend, such that the lines lie above a common horizontal line $\ell$ and have exactly one endpoint on $\ell$. We show that the problem of finding a  maximum clique in a grounded 1-bend string graph is APX-hard, even for strictly $y$-monotone strings. For general 1-bend strings, the problem remains APX-hard even if we restrict the position of the bends and end-points to lie on at most three parallel horizontal lines.  
 We give fast algorithms to compute a maximum clique  for different subclasses of grounded segment graphs, which are formed by restricting the strings to various forms of $L$-shapes.
\end{abstract}

\section{Introduction}\label{se:intro}

A \emph{geometric intersection graph} consists of a set of geometric objects  representing the nodes of the graph, where  two nodes are adjacent if and only if the corresponding objects intersect. Intersection graphs that arise from the intersection of \emph{strings}, i.e., simple curves in $\mathbb{R}^2$, are called  \emph{string graphs}.  
A number of restrictions on the strings have been examined in the literature. Outerstring graphs and grounded string graphs are two  widely studied classes of graphs that resulted from such restrictions. An \emph{outerstring graph} is a string graph, where the strings lie inside a disk, with one endpoint on the boundary of the disk. In a \emph{grounded string} graph, the strings lie above a common horizontal line $\ell$, with one  endpoint on $\ell$. The line $\ell$ is referred to as a \emph{ground line}. 

Although the outerstring graphs and grounded string graphs are the same for general strings, they can be different when we put restrictions on the strings. For example, if we restrict the strings to be straight line segments, the resulting grounded segment graph class is a proper subclass of outersegment graph class~\cite{DBLP:journals/jgaa/CardinalFMTV18}. Strings are often modeled  with polygonal chains, where a \emph{$k$-bend string} is a polygonal chain with at most $k$ bends or $(k+1)$ segments.

While the maximum independent set problem is NP-complete for string graphs (even when the strings are straight line segments), it is polynomial-time solvable for outerstring graphs with a given  outerstring representation of polynomial size~\cite{keil2017algorithm}. Therefore, it is natural to explore other common optimization problems for outerstring graphs. In this paper we explore the  \emph{maximum clique} problem, i.e., we seek a largest subset of  pairwise intersecting  strings. Cabello et al.~\cite{cabello2012clique} proved the maximum clique problem to be  NP-Hard even for ray intersection graphs. Since one can enclose   a ray intersection graph inside a circle such  that all the rays hit the perimeter, their result implies NP-hardness for computing a maximum clique in an outersegment graph. Therefore, an interesting question that arises in this context is whether the maximum clique problem remains NP-hard for grounded segment graphs. While the problem remains open, in this paper, we show NP-hardness for two subclasses of grounded 1-bend  string graphs.

\subsection{Related Research}

The hardness of independent set and maximum clique problems in general graphs inspired researchers to examine these problems for restricted intersection graph classes. Both problems remain polynomial-time  solvable for \emph{circle graphs}, i.e., the intersection of a set of chords of a circle~\cite{Nash10,tiskin2015fast}.
However, both become 
NP-complete in general segment intersection graphs~\cite{kratochvil1990independent}.  

A set of curves is \emph{$k$-intersecting} if every pair of curves have at most $k$ points in common. A string graph is $k$-intersecting if it is the intersection graph of a $k$-intersecting set of curves. Fox and Pach~\cite{FoxP11} gave subexponential time algorithms for computing a maximum independent set for string graphs, as well as algorithms for approximating independent set and maximum clique in $k$-intersecting graphs. 
 
Middendorf and Pfeiffer~\cite{MiddendorfP92} showed the maximum clique  problem to be NP-hard for axis-aligned 1-bend strings, even when the strings are of two types:    \testletter{L} and \textsf{L}. They also showed the problem to be polynomial-time solvable for two cases: (a) For the strings of type $\mathsf{\Gamma}$ and \testletter{L}, and (b) for grounded segments when the free endpoints of the segments lie on a fixed number of horizontal lines. 
  The recognition problem of the intersection graphs of \textsf{L}-shapes is NP-hard~\cite{thesis}.  Throughout the paper we  assume that an intersection representation is given.



Keil et al.~\cite{keil2017algorithm} examined the maximum independent set problem for outerstring graphs. Given an  outerstring representation, where each segment is represented as a polygonal chain, they showed how to compute a maximum independent set in $O(s^3)$ time, where $s$ is the total number of segments in the representation. Bose et al.~\cite{mondal2019} gave an $O(n^2)$-time algorithm when the strings are  
$y$-monotone polygonal paths of constant length with segments at integral coordinates. They also showed this to be the best possible under  strong exponential time hypothesis. 


A rich body of research examines the recognizability of various classes of string graphs~\cite{matousek2014intersection,DBLP:journals/jgaa/CardinalFMTV18,pergel2016edge}. Throughout this paper, whenever we examine an intersection model, we assume that the input graph comes with a   representation satisfying that intersection  model.


\subsection{Contributions}
We first prove that the problem of computing a maximum clique in a grounded 1-bend string graph is APX-hard, even  for strictly $y$-monotone strings. 
 We then show that the problem  remains NP-hard when the bend and end  points of the strings (not necessarily $y$-monotone) are restricted to lie on three horizontal lines.
 Finally, we give fast polynomial-time algorithms for some restricted grounded 1-bend string graphs. In particular, when the grounded 1-bend strings are  1- and 2-sided $L$-shapes. The results are summarized in the following table. Note that  the class of 2-sided grounded $L$-shapes is known to be a proper subclass of grounded segment graph class~\cite{Jelinek019}. Therefore, our results do not settle the time complexity question for computing a maximum clique in a grounded segment graph, and it remains open. 
 We also show the separation between various graph classes.



\begin{table}[ht]
\label{tab1}
\centering
\begin{tabular}{|l|c|c|}
\hline
Graph Class  &  Complexity & Reference\\\hline
\hline
Grounded 1-bend string graphs  & APX-hard  & Section~\ref{apx}\\\hline
Grounded 2-sided $L$-shape graphs & $O\left(\frac{n^2\log^2 n}{(\log \log n)^2}\right)$  & Section~\ref{2sl}\\\hline
Grounded 2-sided square-$L$ graphs  & $O(n^2\log \log  n)$ & Section~\ref{2ssl}\\\hline
Grounded 1-sided $L$-shape graphs  & $O(n^2\log \log n)$  & Section~\ref{1sl}\\\hline 
\end{tabular}
\end{table}

Throughout the paper, we use the term \emph{$L$-shape} to denote  an axis-aligned  1-bend string. An $L$-shape intersection representation is \emph{1-sided} (Figure~\ref{fig:fig1}(a)--(b)), if the $L$-shapes in the representation all turn clockwise,  or all turn anticlockwise. Otherwise, the representation is \emph{two-sided} (Figure~\ref{fig:fig1}(d)).  In a \emph{square $L$-shape representation}, the horizontal and vertical segments of every $L$-shape are of the same length  (Figure~\ref{fig:fig1}(c)). 

\begin{figure}[ht]
\centering
\includegraphics[width=\linewidth]{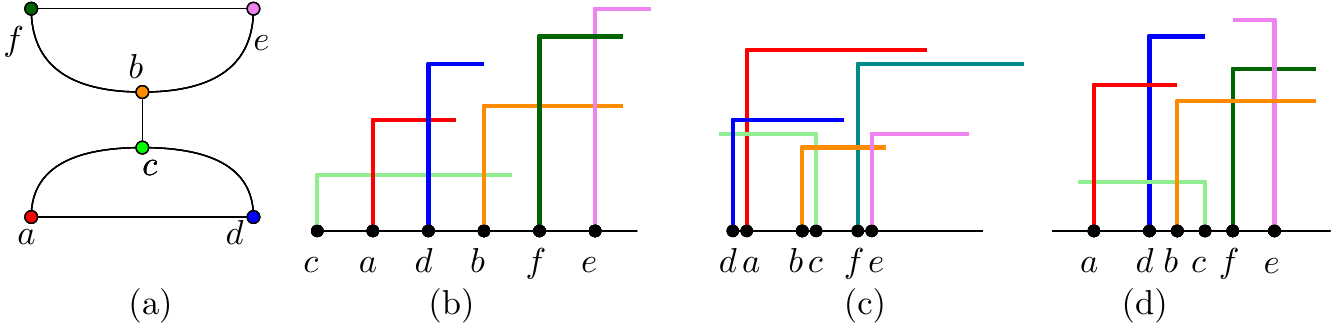}
\caption{(a) A graph $G$. (b) An 1-sided  $L$-shape representation. (c) A square $L$-shape representation. (d) A two-sided $L$-shape representation.}
\label{fig:fig1}
\end{figure}


\section{APX-hardness}
\label{apx}
In this section we show that finding a maximum clique in a grounded 1-bend string graph is an APX-hard problem, even when each string is strictly $y$-monotone. We reduce the maximum independent set  problem in $2k$-subdivisions of cubic graphs, which is APX-hard for any fixed $k\ge 0$~\cite{ChlebikC07}. Here a \emph{$t$-subdivision} of a graph is obtained by replacing each edge $(u,v)$ of $G$ with a path $(u,d_1,d_2,$ $\ldots,d_t,v)$ of $t$ division vertices.  

Let $G$ be 2-subdivision of a cubic graph, we first compute the complement graph $\overline{G}$ and then show how $G$ can be represented as a strictly $y$-monotone grounded 1-bend string graph. Assume that there exists a $(1-\varepsilon)$-approximation algorithm $\mathcal{A}$ for the maximum clique problem in  grounded 1-bend string graphs. Let $C$ be a maximum clique obtained from $\overline{G}$ using $\mathcal{A}$. Let $I^*$ and $C^*$ be a maximum independent set and a maximum clique in $G$ and $\overline{G}$, respectively. Then  $|C| \ge (1-\varepsilon)|C^*|$. Note that an  independent set in a graph corresponds to a clique in its complement, and vice versa. Therefore, we obtain   an independent set of size at least $(1-\varepsilon)|I^*|$ in $G$, which contradicts the APX-hardness for computing a maximum independent set in $G$.  Therefore, it now suffices to show that $\overline{G}$ admits a strictly $y$-monotone grounded 1-bend string representation.

\begin{figure*}[ht]
    \centering
    \includegraphics[width=\textwidth]{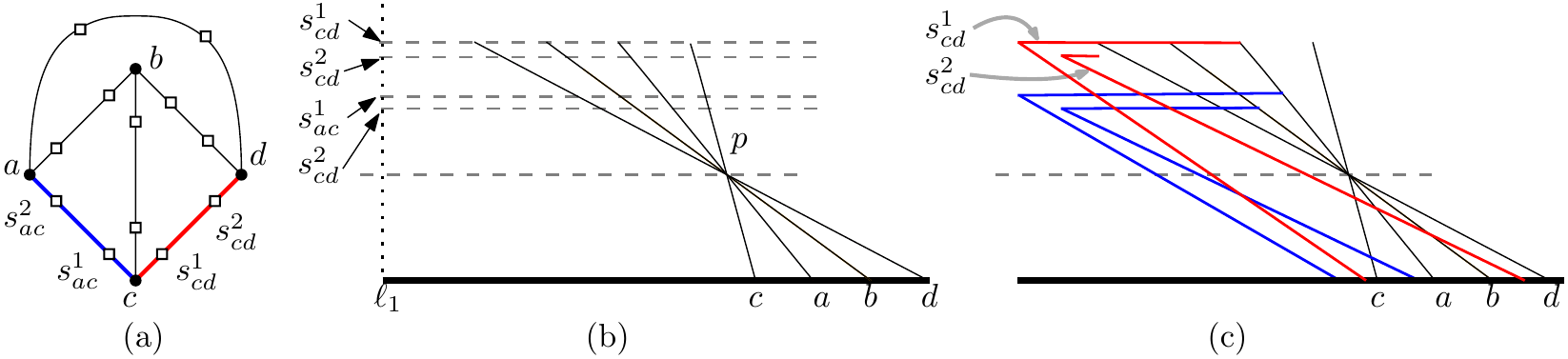}
    \caption{Illustration for the (a) graph $G$, and (b) the construction of the original strings. (c) Construction of a  grounded 1-bend string representation for $\overline{G}$.}
    \label{fig:x}
\end{figure*}


\subsection{Grounded 1-bend string representation with $y$-monotone  strings}

Assume that $G$ is the 2-subdivision of a cubic graph  $H$ (Figure~\ref{fig:x}(a)). Let the $x$-axis be the ground line. While constructing the representation for $\overline{G}$, we  will refer to the strings corresponding to the vertices that originally belong to  $H$ as the \emph{original strings}, and the other strings as the \emph{division strings}. 

Note that the original vertices form a clique in $\overline{G}$. For each vertex $i=1\ldots,n$ in $H$, we construct an original string, which is a straight line segment with end-points $(i,0)$ and  $(-3ni,6n+2)$. All these lines  pass through the point $p=(0,2)$. Consequently, we obtain a set of pairwise intersecting strings (Figure~\ref{fig:x}(b)).

We now describe the construction for the division strings. Let $\ell_1$ be the vertical segment that starts at   $(-3n^2-1,6n+2)$ and hits the ground line (Figure~\ref{fig:x}(b)). 
Then there are $6n$ horizontal lines (through integral coordinates)  between $y=3$ and $y=6n+6$. This number is larger than $3n$, i.e., the number of division vertices in $\overline{G}$. We use these lines to create the division strings. 
 
We number the edges of $H$ from  $1$ to $(3n/2)$. Let $(c,d)$ be the $j$th edge in $H$ and let $s_{cd}^1,s_{cd}^2$ be the corresponding  division vertices in $G$. The vertex $s_{cd}^1$ (resp., $s_{cd}^2$) is adjacent to all the original  
vertices and division vertices of $\overline{G}$, except for $c$ (resp., $d$) and $s_{cd}^2$ (resp., $s_{cd}^1$). We now construct a pair of division strings that realize these adjacencies.   

Let $c$ and $d$ be the $t$th and $r$th  vertex of $H$. Assume without loss of generality that the string of $c$ appears to the left of the string of $d$ on the ground line. 
The division string for $s_{cd}^1$ starts at the ground line,  makes a bend at $\ell_1$, and then intersect all the original strings that lie to the left of the string for $c$. In particular, the string starts at $(t-\frac{1}{j},0)$, makes a right turn at $q=(-3n^2-1,6n-j+3)$ following the line $y=(6n-j+3)$, and stops as soon as it intersects all the strings that appear before the string of $c$ on this line.  
The division string for $s_{cd}^2$ starts at $w=(r-\frac{1}{j},0)$, makes a right turn at the intersection point of the lines $qw$ and $y=(6n-j+2)$, and continues until  it intersects all the strings that appear before the string of $d$ (Figure~\ref{fig:x}(c)). 

Since the permutation of the original strings on the ground line is opposite to the permutation on the line $y=6n+2$, it is straightforward to verify that the string for $s_{cd}^1$ (resp., $s_{cd}^2$) intersects all the original strings except the one for $c$ (resp., $d$).

By construction, the strings of $s_{cd}^1$ and $s_{cd}^2$ are disjoint and intersect all the previously added division strings. Otherwise, suppose for a contradiction that a  previously added division string $z$  has not been intersected by the string for $s_{cd}^2$. Then  $z$ should appear to the left of $s_{cd}^2$. Thus the segment of $z$ that touches the ground line will be almost vertical. Therefore, $z$ will have a  negative $x$-coordinate at the ground line, which contradicts the property that every division string starts with a positive $x$-coordinate. The argument is same for $s_{cd}^1$. 
 This completes the grounded string representation for $\overline{G}$.

Since the coordinates explicitly described above are polynomial in $n$, the intersection of the line segments required to carry out the construction also have coordinates of polynomial size. Hence one can compute the string representation in polynomial time. Note that the strings that we computed are $y$-monotone. To make the strings strictly $y$-monotone, one can carry out the same construction for division strings with a set of slanted parallel lines  of small positive slope, instead of horizontal ones. We thus have the following theorem. 

\begin{theorem}
The maximum clique problem is APX-hard for  grounded 1-bend string graphs, even when the strings are strictly $y$-monotone. 
\end{theorem}

\subsection{Grounded string representation on a few lines}

Given a 1-bend string, we will refer to its two endpoints as  \emph{fixed} or \emph{free} depending on whether they lie on the ground line or not. Similarly, we call its segments as  \emph{fixed} or \emph{free} depending on whether the segment is adjacent to the ground line or not. 

We slightly modify the construction of the previous section such that the bends and end-points of the strings lie on at most three lines (above the ground line): $\ell_1,\ell_2,\ell_3$, as illustrated in Figure~\ref{fig:y}. We omit the coordinate details for this construction. It is straightforward to use a very similar approach as in the previous section to compute the representation with coordinates of polynomial size. 

We create the original strings such that bend-points and free endpoints lie on  $\ell_2$ and $\ell_1$, respectively. We also ensure that the  order of the free endpoints on $\ell_1$ is opposite to the order of the fixed  endpoints.

Let $q$ be a point to the right of all the free endpoints of the original strings on $\ell_1$. A pair of division strings start at the ground line near their corresponding original strings (as in our earlier construction), but their bends  are placed  on $\ell_2$ and $\ell_3$ such that the strings remain disjoint. We also ensure that fixed segments of these strings lie to the right of $q$. Consider now a division string $z$ that starts at the ground line and reaches $\ell_2$ or $\ell_3$. If all the required intersections are realized at its fixed segment, then we create the free segment by connecting $q$ and its bend-point. If only a subset  of the required intersections is realized at its fixed segment, then the free segment is created by connecting the bend-point to  an appropriate point $q'$ to the left of $q$ on $\ell_1$. Since the permutation of the fixed endpoints of the original strings is  the reverse of the permutation of their free endpoints, such  a point $q'$ must exist.

\begin{figure}[h]
    \centering
    \includegraphics[width=.75\linewidth]{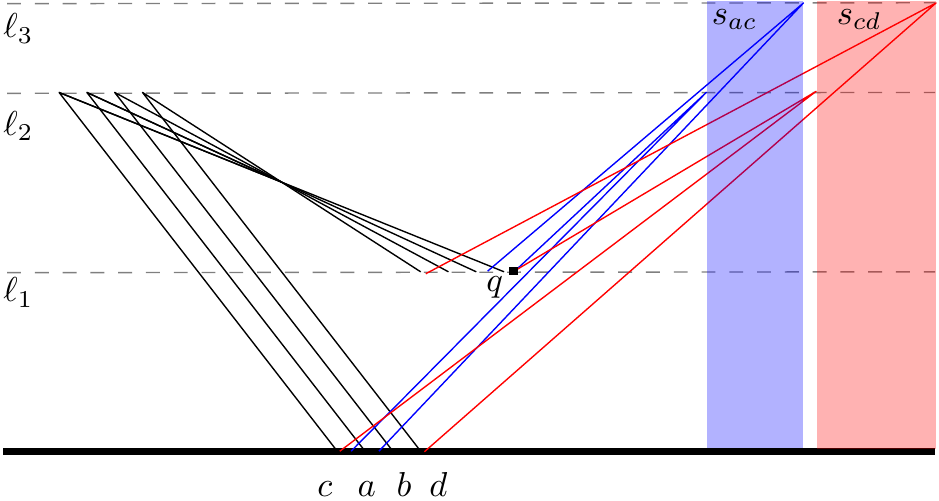}
    \caption{  Construction of a  grounded string representation for $\overline{G}$ on three lines.}
    \label{fig:y}
\end{figure}

The division vertices are created in pairs and their bend-points are placed to the right of all the previously created bend-points, as illustrated using the vertical stripes in  Figure~\ref{fig:y}. Note that every newly created string $z$ needs to reach either $q$ or to a point to the left of $q$ on $\ell_1$. Therefore, we can choose the new bend-point of $z$ sufficiently far apart such that its free segment crosses all the previously added division strings at their fixed segments. This completes the required construction for $\overline{G}$. We thus have the following theorem.

\begin{theorem}
The maximum clique problem is APX-hard for  grounded 1-bend string graphs, even when the bends and end-points are restricted to lie on three horizontal lines. 
\end{theorem}

\section{Two-sided $\mathbf{L}$-shapes}
\label{2sl}

In this section we consider the case when the grounded strings are two-sided $L$-shapes. We use dynamic programming to compute a maximum clique on this class of graphs, and give an $O(n^3)$-time algorithm to compute a maximum clique. We then improve the running time to $O((\frac{n\log n}{\log \log n})^2)$. We assume that all the $L$-shapes are in general position, i.e., no two segments in the intersection representation lie on the same horizontal or vertical line.  

\subsection{Technical Details}

Let $P$ be a set of points in $\mathbb{R}^2$, where each point is assigned a positive weight. 
An \emph{orthogonal range maxima query} on $P$ provides an axis aligned rectangle and returns a point with the maximum weight.  

For small point sets in a rank space, there exists a data structure that support queries and weight updates in $O(1)$ time, as stated in the following lemma. The idea is to reduce a query on $P$ into a constant number of  queries on various sets, where each set  contains at most $t=\frac{\log n}{6\log \log n}$ points. The solution to all possible $O(\log ^2 n)$ queries for all possible sets that can arise from at most $t$ points are precomputed and stored in a table.  Hence a query can be answered in constant time, and an update can be done in constant time by maintaining a pointer that keeps track of the set that represents the current set of points.  

\begin{lemma}[Nekrich~\cite{DBLP:conf/compgeom/Nekrich20}]
\label{l1}
If a point set $P$ contains an  $O(\log^{2\varepsilon} n)$ points, where $0<\epsilon<1$, and coordinates and weights  of all points are bounded by $O(\log n)$, 
then we can support range maxima queries under  weight updates on $P$ 
 in $O(1)$ time.
\end{lemma} 
  
For the general point set, we modify the range tree data structure to obtain an $O((\frac{\log n}{\log \log n})^2)$ time query and weight update.

\begin{lemma}
\label{l3}
Let $P$ be a set of $n$ points with all the coordinates and weights bounded by $O(n)$. Then there exists a data structure that supports range maxima queries and weight updates on $P$ in   
$O((\frac{\log n}{\log \log n})^2)$ time.
\end{lemma}
\begin{proof}
 We construct the range tree $T_y$ with node degree $\log^\varepsilon n$ on the $y$-coordinates of the points. Let $w$ be a node of $T_y$ and let $S_w$ be the set of points that is stored at the leaves of subtree rooted at $w$. We then store another range tree $T^w_x$ with node degree $\log^\varepsilon n$ on the $x$-coordinates of points $S_w$. We refer to $T^w_x$ as a second-level tree.
 
 Since the height of $T_y$ is $O(\log n/\log \log n)$, a query (or an update) selects $O(\log^{1+\varepsilon} n)$ nodes 
 that contain the  points with correct $y$-coordinates. These nodes can be divided into $O(\log n/\log\log n)$ groups, such that all nodes in the same group are siblings. Consider such a group of siblings and their common parent $w$. At each such node $w$, we need to query $T^w_x$ to identify $O(\log n/\log \log n)$ nodes with correct $x$-coordinates. Hence we have  $O((\frac{\log n}{\log \log n})^2)$ candidate nodes in the second-level trees. To achieve a  $O((\frac{\log n}{\log \log n})^2)$ query (or update) time, we now show how each selected nodes in the second-level trees can be processed in $O(1)$ time.
 
 First, at each second-level tree $T^w_x$, we modify the $y$-coordinates of the points so that they are bounded by $O(\log ^\varepsilon n)$. To achieve this, for each point $p$, we find the child of $w$ that store $p$. If $p$ is stored at the $j$th child of $w$, then we replace the $y$-coordinate of $p$ with $j$. Since the degree is bounded by $O(\log ^\varepsilon n)$, the $y$-coordinates of the points of the second level trees are now bounded by $O(\log ^\varepsilon n)$. Second, at each $u$ node of  $T^w_x$, we select $O(\log ^{2\varepsilon} n)$ points and store them in a data structure of Lemma~\ref{l1}. These are the points with the maximum weight for each $y$-coordinate and each child of $u$. Since there are $O(\log ^{\varepsilon} n)$ distinct $y$-coordinates and $O(\log ^{\varepsilon} n)$ children, the number of points is $O(\log ^{2\varepsilon} n)$. Note that the $y$-coordinates of these points are already bounded by $O(\log n)$. We map the $x$-coordinate of each selected point $q$ to the index of the child of $u$ that stores $q$. Hence the $x$-coordinates become bounded by $O(\log ^{\varepsilon} n)$. Since there are at most $O(\log^{2\epsilon} n)$ points, the weight can be reduced to a rank space. We set $\varepsilon\le  1/2$ such that the rank is bounded by $O(\log n)$. Hence we can store these points in a  data structure of Lemma~\ref{l1}. 
 
For each of the $O((\frac{\log n}{\log \log n})^2)$ queries (or updates)  at the second-level trees, we can now use the data structure of Lemma~\ref{l1}. Hence the overall running time becomes  $O((\frac{\log n}{\log \log n})^2)$. 
\end{proof}

\subsection{Computing a Clique}

Let $G$ be an intersection graph of two-sided $L$-shapes, and let $Q$ be a maximum clique with at least two vertices in $G$. Let $a$ and $b$  be the highest and second-highest $L$-shape in $Q$, respectively, and without loss of generality assume that $a$ appears to the left of $b$ on the ground line (Figure~\ref{fig:fig2}).  Then any other $L$-shape $c$ in $Q$ must be below the line $\ell$ determined by the horizontal segment of $b$. Furthermore, since $c$ intersects both $a$ and $b$, its endpoint on the ground line must be to the left of $a$ or to the right of $b$. In other words, the  interval $[a,b]$ on the ground line acts as a forbidden interval for the other vertices in the clique. 

If $c$ is the next highest clique after $b$ in $Q$, then depending on whether it lies to the left or right of the interval $[a,b]$,  the forbidden region for the remaining vertices in $Q$ grows to either $[c,b]$ or $[a,c]$. Without loss of generality assume that $c$ lies to the right of $b$, and the new forbidden region is $[a,c]$. Then an $L$-shape intersecting $c$ and $a$ must also intersect $b$, where $b$ already belongs to $Q$. One can thus continue adding a new $L$-shape that intersect the $L$-shapes representing the current forbidden interval, without worrying  about the $L$-shapes which have been chosen already. We use this idea to design the dynamic programming algorithm. 

If $G$ does not contain any edge, then the maximum clique size is 1. Otherwise, let $D(a,b)$ denote a maximum clique, where $a,b$, or $b,a$ are the first and second highest $L$-shapes, and $a$ lies to the left of $b$. Let $c$ be an $L$-shape that intersects both $a$ and $b$, and for an $L$-shape $w$, let $w_x$ be its $x$-coordinate on the ground line. Then  
\[
D(a,b) =
\begin{dcases*}
2 \text{, if $c$\;doesn't\;exist\,,} & \\
   \max
\begin{dcases*}
\max_{c_x<a_x} \{D(c,b)\} +1
   & \\[1ex]
\max_{c_x>b_x} \{D(a,c)\}+1,
   &   otherwise.
\end{dcases*}
\end{dcases*}
\]

\begin{figure}[ht]
\centering
\includegraphics[width=.6\textwidth]{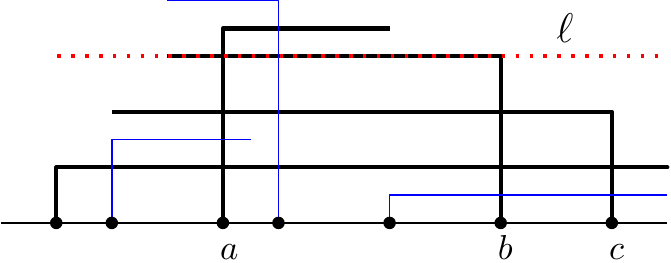}
\caption{Illustration for the dynamic programming for two-sided $L$-shapes. The clique $Q$ is shown in black.}
\label{fig:fig2}
\end{figure}

We take the maximum over all pairs of $L$-shapes $a,b$ in the input. We can use a 2-dimensional table $T(a,b)$ to store the solution of $D(a,b)$. The size of the dynamic programming table is $O(n^2)$, where $n$ is the number of $L$-shapes. Computing each entry of the table requires $O(n)$ table look-up. Therefore, it is straightforward to compute the maximum clique in $O(n^3)$ time. 

We now show how this computation can be done faster by using the range maxima data structure of Lemma~\ref{l3}. Assume that we are computing an entry $D(a,b)$ where $a$ and $b$ are the first and second highest $l$-shapes, respectively.  For each $L$-shape $v$, we maintain a data structure  $S_v$  that maintains a set of weighted points as follows. Each point $p\in S_v$  corresponds to an $L$-shape $w$ that intersects $v$. The $x$ and $y$-coordinates of $p$ are the $x$ and $y$-coordinates of the bend point of $w$, and the weight of $p$ is the solution of $D(v,w)$ or $D(w,v)$ depending on whether $v$ precedes $w$ on the ground line. 

To compute $D(a,b)$, it now suffices to compute two range maxima query (Figure~\ref{fig:fig2}), as follows. One range is defined by the vertical segment of $b$, the horizontal line $\ell$ determined by the highest point of $b$, and the ground line. The other range is defined by the vertical segment of $a$, the horizontal line $\ell$ determined by the highest point of $b$, and the ground line.  

Since a range maxima query takes $O((\frac{\log n}{\log \log n})^2)$ time  (Lemma~\ref{l3}), the dynamic programming table can be filled in $O(n^2(\frac{\log n}{\log \log n})^2)$ time. The following theorem summarizes the result of this section.

\begin{theorem} 
Given a set of $n$ grounded two-sided $L$-shapes, one can compute  a maximum clique in the corresponding intersection graph in $O((\frac{n\log n}{\log \log n})^2)$ time. 
\end{theorem}

\section{ Two-sided  Square $L$-shapes}
\label{2ssl}
In this section, we consider two-sided square $L$-shapes, and give an $O(n^2\log \log n)$-time algorithm to compute a maximum clique. We assume the $L$-shapes are in general position.  

\begin{figure}[h]
\centering
\includegraphics[width=\linewidth]{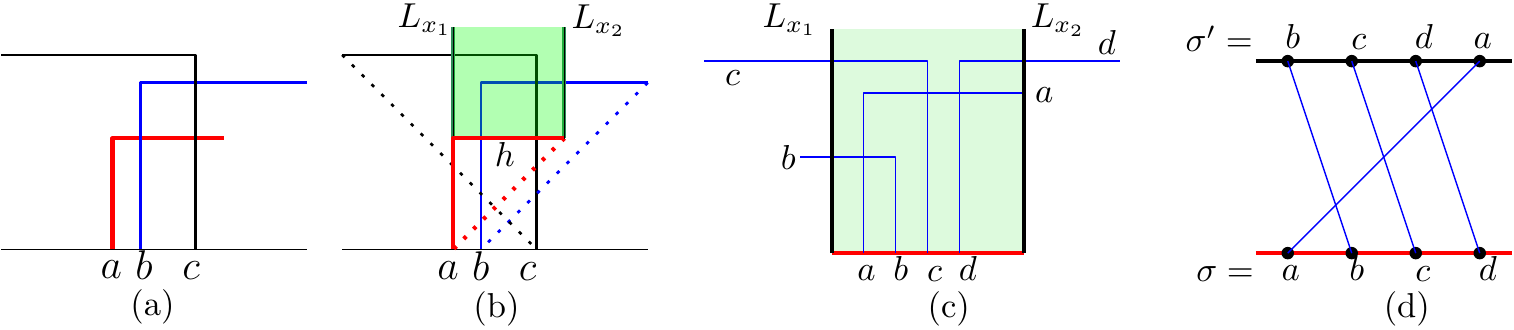}
\caption{ (a) A maximum clique in a two-sided square $L$-shape representation. (b) Illustration for the properties of the clique.}
\label{fig:fig5}
\end{figure}

We first discuss some geometric properties of a maximum clique, which will help us to design the dynamic programming (Figure~\ref{fig:fig5}). Let $a$ be the lowest $L$-shape of a maximum clique $Q$. Then all other $L$-shapes must intersect the horizontal segment $h$ of $a$. Let $L_{x_1}$ and $L_{x_2}$ be vertical lines through the left and right endpoints of $h$. We now have the following observation.
\begin{lemma}
\label{lem:circ}
Every square $L$-shape that intersects $h$, must intersect $L_{x_1}$ or $L_{x_2}$.
\end{lemma}
\begin{proof}
 Suppose for a contradiction that $b$ is an $L$-shape that intersects $h$ but does not intersect $L_{x_1}$ or $L_{x_2}$.  Then the length of the horizontal segment of $b$ is at most the length of $h$. Thus the maximum length for the vertical segment of $b$ is also bounded by the length of $h$. Since the bend-point of $b$ is above $h$,  the vertical segment of $b$ cannot reach the ground line, which contradicts that $b$ is a grounded square $L$-shape. 
\end{proof}

Let $R$ be the region above $h$, bounded by $L_{x_1}$ and  $L_{x_2}$. Let $H$ be the intersection graph induced by the $L$-shapes intersecting $h$. 
In the following we show that $H$ is a permutation graph. 
Here a \emph{permutation graph} is an intersection graph of straight line segments where the line segments are drawn between a pair of parallel lines, i.e., all the top endpoints of these segments lie on one line and the bottom endpoints lie on the other.  
\begin{lemma}
The graph $H$ is a permutation graph.
\end{lemma}
\begin{proof}
Let $\sigma$ be the labels of the $L$-shapes that intersect $h$ in left to right order, and let $\sigma'$ be the labels of these $L$-shapes in the clockwise order of their intersection points on the boundary of $R$ excluding $h$. Let $H'$ be the permutation graph determined by the orderings $\sigma$ and $\sigma'$  on two parallel lines, as illustrated in  Figure~\ref{fig:fig5}(c)--(d). It now suffices to prove that two vertices are adjacent in $H$ (i.e., the corresponding $L$-shapes intersect) if and only if they are adjacent in $H'$ (i.e., the corresponding line segments  intersect).   

Let $a$ and $b$ be two $L$-shapes that intersect in the square-$L$ representation of $H$. If they both intersect $L_{x_1}$ (resp., $L_{x_2}$), then their ordering in $\sigma$ is different from that of $\sigma'$. If one intersects $L_{x_1}$ and the other intersects $L_{x_2}$, then again, their ordering in $\sigma$ is different from that of $\sigma'$. 
Thus in both cases, the corresponding vertices in $H'$ must be adjacent, i.e., the line segments must intersect in the permutation graph representation of $H'$.

Consider now the case when the $L$-shapes $a$ and $b$ do not intersect. If they both both intersect $L_{x_1}$ (resp., $L_{x_2}$), then on $L$-shape must `nest' the other and thus  their ordering would be the same in $\sigma$ and $\sigma'$. If one $L$-shape intersects $L_{x_1}$ and the other intersects $L_{x_2}$, then the one that intersects $L_{x_1}$ must appear to the left of the other on the ground line. Therefore, in both cases, the corresponding vertices in $H'$ cannot be adjacent, i.e., the line segments cannot intersect in the permutation graph representation of $H'$.
\end{proof}

We precompute a permutation representation of the permutation graph for every $L$-shape. A maximum clique in a permutation graph corresponding to an $L$-shape determines a maximum clique containing that $L$-shape. Since a maximum clique in a permutation graph can be obtained in  $O(n\log  \log n)$ time~\cite{YuTC93}, we can find the maximum clique containing an $L$-shape in the same time. 
 
Finally, we iterate the above process over all $L$-shapes to find the maximum clique. Let $S$ be the input $L$-shapes. It is straightforward to precompute the permutation representation of the permutation graph for every $L$ shape in $O(n^2)$ time in total by first  computing the sorted orders of the vertical segments and horizontal segments of the  $L$-shapes, and then for each segment $a$, computing the $L$-shapes that intersects $h$, $L_{x_1}$ and $L_{x_2}$ in sorted order.  Hence the time complexity for computing maximum clique is $O(n^2) + \sum_{q \in S} O(n\log \log n) \in O(n^2\log \log  n)$, where $S$ is the set of $L$-shapes in the input.

The following theorem summarizes the result of this section.

\begin{theorem}
Given a set of $n$ grounded two-sided square $L$-shapes, one can find the maximum clique of the corresponding intersection graph in $O(n^2\log \log  n)$ time.
\end{theorem}

\section{One-sided $L$-shapes}
\label{1sl}
In this section, we consider 1-sided  $L$-shapes, and give an $O(n^2\log \log n)$-time algorithm to compute a maximum clique. We assume the $L$-shapes are in general position.  We noticed that Chmel~\cite{thesis} has  independently discovered this result at the same time the conference version of this paper was accepted. Chmel mentioned an $O(n^2 \log n)$ time complexity.

\begin{figure}[ht]
\centering
\includegraphics[width=.5\textwidth]{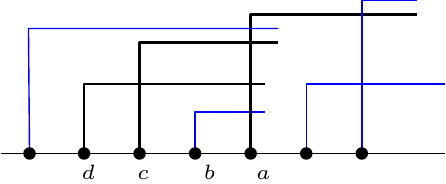}
\caption{Illustration for the dynamic programming for 1-sided $L$-shapes. The clique $Q$ is shown in black.}
\label{fig:ones}
\end{figure}


Let $S$ be the set of $L$-shapes  in the input and let $Q$ be a maximum clique. Let $a$ be the highest $L$-shape in $Q$. Then all the other $L$-shapes in $Q$ must intersect the vertical segment of $a$ (Figure~\ref{fig:ones}).  
There may also exist $L$-shapes (e.g., $b$) that do intersect the vertical segment of $a$ but does not belong to $Q$. 

Let $D(S)$ be the maximum clique in the intersection graph of the set $S$ of input $L$-shapes. For any $L$-shape $q\in S$, let $N(q)$ be the subset of $S$ such that every $L$-shape in $N(q)$ intersects the vertical segment of $q$. Then $D(S)$ can be defined as follows. 
\[ 
D(S)= \max_{q \in S} (D(N(q))+1)
\]

To compute the maximum clique efficiently, we do some preprocessing. We first compute the intersection graph $G$ of $S$ in $O(n^2)$ time. We then compute a sorted list $S_x$ consisting of the fixed endpoints (on the grounded line) of all the $L$-shapes. We next compute another sorted list $S_h$ of the heights of the $L$-shapes (Figure~\ref{fig:list}). Both these lists can be computed in $O(n \log n)$ time. Hence the total preprocessing takes $O(n^2)+O(n\log n) \in O(n^2)$ time. 

\begin{figure}[ht]
\centering
\includegraphics[width=.6\linewidth]{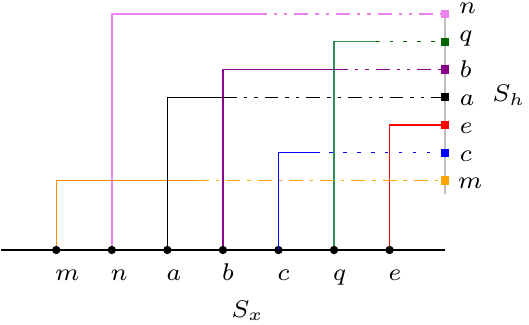}
\caption{ Illustration for $S_x$ and $S_h$.}
\label{fig:list}
\end{figure}

For each $q \in S$, we now compute the maximum clique in $N(q)$. First note that $N(q)$ corresponds to a permutation graph, where the edges are determined by the intersection of the $L$-shapes in $N(q)$.  We first find two ordered lists for the $L$-shapes of $N(q)$, one list corresponds to the order they appear on the ground line, and the other list corresponds to the order they appear on the  vertical segment of $q$. We  then  use an $O(n\log  \log n)$-time algorithm~\cite{YuTC93} for computing a maximum clique in a permutation graph to compute the maximum clique in $N(q)$. 

To list the $L$-shapes of $N(q)$ in the order they appear on the ground line, we scan $S_x$ from left to right and create a new ordered list $S'_x$ with the $L$-shapes that intersect the vertical segment of $q$. We then scan $S_h$ and create a new ordered list $S'_h$ that contains the $L$-shapes intersecting the vertical segment of $q$. Constructing $S'_x$	and $S'_h$ takes $O(n)$ time. Hence computing $N(q)$ and finding a maximum clique in $N(q)$ takes $O(n)+O(n\log \log n)$ time. Thus the total running time is
 \[ 
 \sum_{q \in S} O(n)+O(n\log \log n) \in O(n^2 \log \log n).
 \]

The following theorem summarizes the result of this section.
\begin{theorem}
Given a set of $n$ grounded 1-sided  $L$-shapes, one can find the maximum clique of the corresponding intersection graph in $O(n^2 \log \log n)$ time.
\end{theorem} 

\section{Separation Between  Graph Classes}


In this section we show that there exist graphs that can be represented with 2-sided square $L$-shapes, but not with grounded 1-sided $L$-shapes (Theorem~\ref{sep1}). Furthermore, there exist graphs that can be represented with grounded 2-sided  $L$-shapes but not with grounded 2-sided  square $L$-shapes (Theorem~\ref{sep2}). 


\begin{theorem}
\label{sep1}
There is a graph $G$ which admits a representation with  2-sided grounded square  $L$-shapes but not with 1-sided grounded $L$-shapes. 
\end{theorem}
\begin{proof}
Jel\'inek and T\"{o}pfer~\cite{Jelinek019} showed that there exist a graph $G$ that can be represented using  grounded 2-sided $L$-shapes but not with grounded 1-sided  $L$-shapes.  It thus suffices to show that $G$ admits an intersection representation with grounded 2-sided square $L$-shapes. 

Figure~\ref{fig:gr}(a) illustrates an intersection graph representation of $G$ with grounded 2-sided $L$-shapes.  The $L$-shapes can be partitioned into five  groups as follows: \textit{Group 1:} Four $L$-shapes $x_1,x_2,x_3,x_4$ with edges ${(x_1,x_4), (x_1,x_2), (x_2,x_3)}$.  
 \textit{Group 2:} Four $L$-shapes $x_5,x_6,x_7,x_8$ with edges ${(x_5,x_8), (x_7,x_8), (x_6,x_7)}$.  
   \textit{Group 3:} Four $L$-shapes $x_{9},x_{10},x_{11},x_{12}$ similar to Group 1.  
    \textit{Group 4:} Four $L$-shapes $x_{13},x_{14},x_{15},x_{16}$ similar to Group 2.
     \textit{Group 5:} A cycle $y_1,\ldots,y_{80}$ with $x_i$ intersecting the $L$-shape corresponding to $y_{5i}$. 
     
The $L$-shapes corresponding to Group 1 and 3 are shown in blue. The $L$-shapes corresponding to Group 2 and 4 are shown in red. The $L$-shapes corresponding to the cycle are shown in black.  
 
Figure~\ref{fig:gr}(b) and (c) illustrate   a representation of Group 1 and Group 2 with square $L$-shapes and their associated $L$-shapes on the cycle. We can scale up and patch the representation of Group 2 with that of Group 1 such that the rightmost $L$-shape intersects the leftmost $L$-shape of Group 2. We can thus  extend the drawing to integrate the Groups 3 and 4. A schematic representation of this configuration is illustrated in Figure~\ref{fig:gr}(d).  The first and the last $L$-shapes on the cycle (i.e., $y_1$ and $y_2$) only intersect the leftmost and rightmost of these configuration. Hence it is straightforward to add them to the configuration completing the square $L$-shape representation of $G$.
\end{proof}
\begin{figure}[hpt]
    \centering
    \includegraphics[width=.8\linewidth]{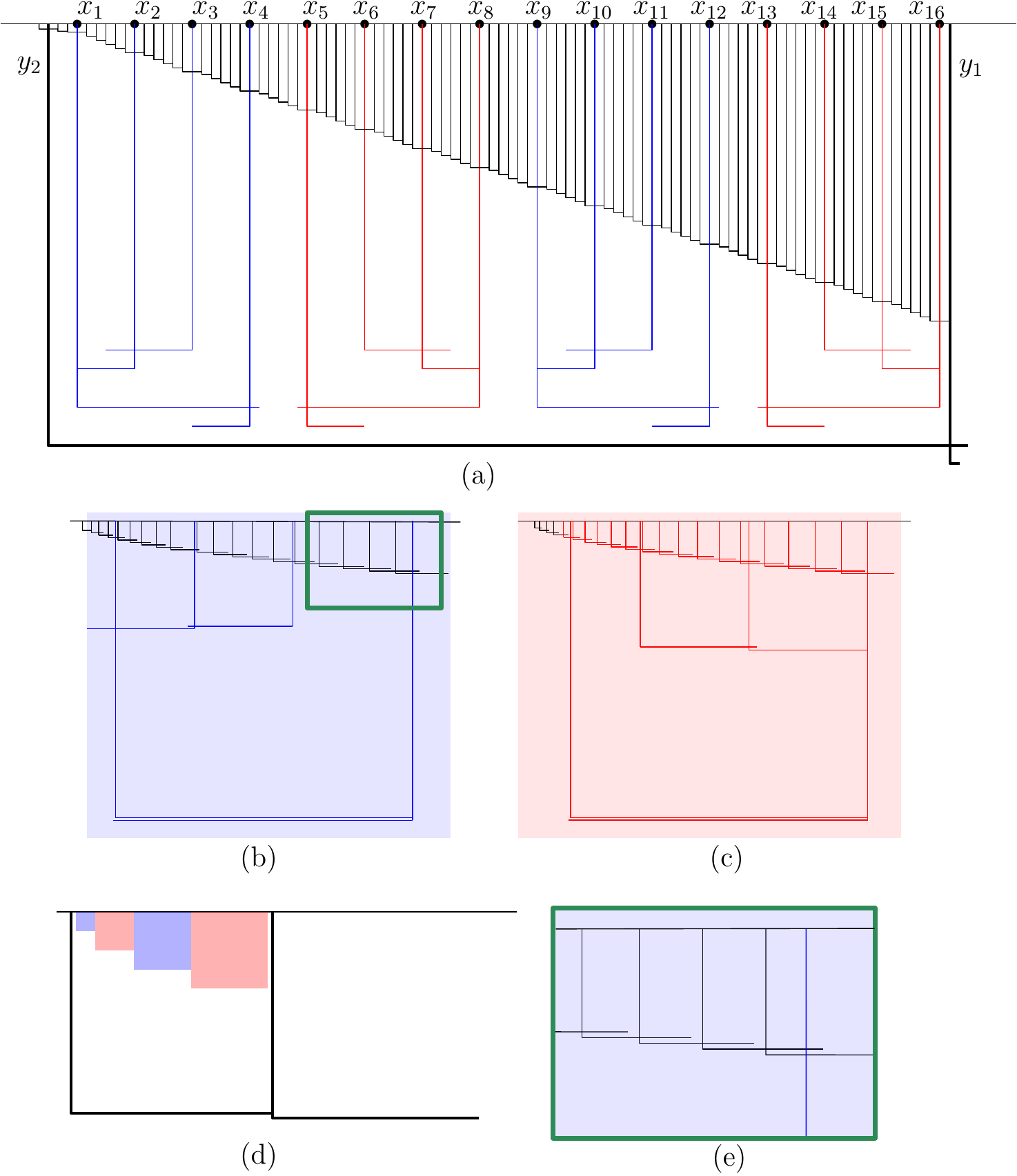}
    \caption{(a) An intersection graph $G$ of grounded 2-sided $L$-shapes that cannot be represented with 1-sided $L$-shapes. (b)--(c) Illustration for an intersection representation of $G$ with grounded 2-sided square $L$-shapes. (e) A zoomed in view of the $L$-shapes of figure (b). }
    \label{fig:gr}
\end{figure}

\begin{theorem}
\label{sep2}
There exists a graph $G$ that admits a representation with  2-sided grounded $L$-shapes but not with 2-sided grounded square $L$-shapes. 
\end{theorem}
\begin{proof}
We will use the wheel graph $G$ of 6 vertices to prove this theorem  (Figure~\ref{fig:case1}(a)). 
 Suppose for a contradiction that $G$ admits an intersection representation $R$ with grounded 2-sided square $L$-shapes. Consider now a representation $R'$ of the 5-cycle $C = (a,b,c,d,e)$ in $R$. We now examine the first and last $L$-shapes of $R'$, and examine two cases depending on whether they are adjacent on  the 5-cycle.\smallskip

\noindent\textbf{Case 1 (The first and last $L$-shapes in $R'$ are adjacent on $C$):} 
Without loss of generality assume that $a$ and $e$ are the leftmost and rightmost $L$-shapes, and $a$ is taller than $e$. Since $b$ is adjacent to $a$, it must be of type  \testletter{L}.
Since $c$ is adjacent to $b$, $c$ can be grounded to the right or left of the $b$. We thus consider two subcases depending on $c$’s position. 

\textit{Case 1.1 ($c$ is lies to the  right of $b$):} Since $d$ needs to intersect $c$ and $e$,  $d$ must lie to the left of $c$ on the ground line and be of type $\mathsf{\Gamma}$. Since $d$ must not intersect $b$ it needs to lie to the right of $b$. Hence $a$, $b$, $c$, $d$ and $e$ have unique x-coordinate order on the ground line.

\begin{figure}[pt]
    \centering
    \includegraphics[width=\linewidth]{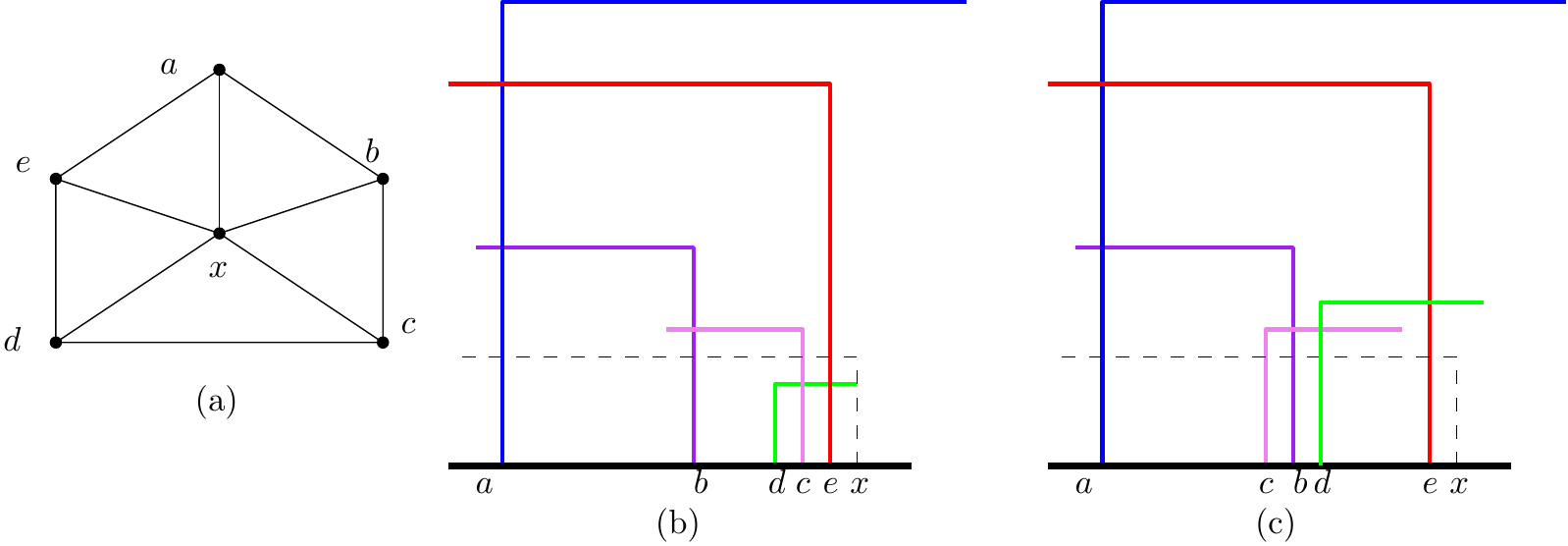}
    \caption{(a) A wheel graph. Illustration for  (b)  Case 1.1 (c) and Case 1.2.}
    \label{fig:case1}
\end{figure}
We now need to add the vertex $x$. 
If we put $x$ between $a$ and $b$, then to intersect both $a$ and $b$, it must be higher than $b$. Hence $x$ would fail to intersect $d$.  
If we put $x$ between $b$ and $e$, then $x$ cannot intersect both $b$ and $e$ simultaneously. 
Therefore, either $x$ lies to the right of $e$ or to the left of $a$. 

If $x$ lies to the left of $a$, then to intersect $d$, it has to be the lowest $L$-shape. Since the height of $c$ is larger than $d$ and since $c$ does not intersect $a$, $x$ will fail to intersect $c$. Consider now the case when $x$ lies to the right of $e$, as illustrated in Figure~\ref{fig:case1}(b). In this case the height of $x$ must be smaller than that of $c$. Since $c$ does not intersect $a$, in the square representation $x$ will fail to intersect $a$.


\textit{Case 1.2 ($c$ is lies to the  left of $b$):} Similar to Case 1.2, this case also results into a unique x-coordinate ordering for the $L$-shapes on the 5-cycle.   Figure~\ref{fig:case1}(c) illustrates this scenario. If $x$ lies to the left of $c$, then the height of $x$ must be smaller than $c$. Since $c$ does not intersect $e$, $x$ will also fail to intersect $e$. If $x$ lies between $c$ and $d$, then to intersect $b$ and $d$, the height of $x$ must be smaller than that of $b$ and $d$. Hence $x$ will not be able to intersect $a$ and $e$ simultaneously. If $x$ lies to the right of $d$, then to intersect both $b$ and $e$, $x$ must lie also to the right of $e$. In this case the height of $x$ must be smaller than $c$. Hence $x$ will fail to intersect $c$.
 

\noindent\textbf{Case 2 (The first and last $L$-shapes  in $R'$ are not adjacent on $C$):} 
Without loss of generality assume that $a$ and $d$ are the first and last $L$-shapes,  and $a$ is taller than $d$. Here $b$ must lie to the right of $a$ and be of type   \testletter{L}. We distinguish two cases based on whether $c$ lies to the left or right of $b$.

\textit{Case 2.1 ($c$ is lies to the  left of $b$):} In this case we can find a unique  x-coordinate ordering of the $L$ shapes on the 5-cycle, as shown in Figure~\ref{fig:case21}(a). However, since $d$ must intersect $c$, the height of $d$ must be smaller than that of $c$. Furthermore, since $c$ does not intersect $e$, $c$ must be of type \testletter{L}. Since $b$ is intersecting $a$, $c$ must intersect $a$, which contradicts that $a$ and $c$ are non adjacent on the 5-cycle.




\textit{Case 2.2 ($c$ is lies to the  right of $b$):} In this case, we obtain a unique x-coordinate ordering for the $L$-shapes on the 5-cycle.   Figure~\ref{fig:case21}(b) illustrates this scenario. If $x$ lies to the left of $b$, then $x$ cannot intersect $a$ and $d$ simultaneously. If $x$ lies between $b$ and $e$, then to intersect $b$, the height of $x$ must be smaller than that of $b$. Hence $x$ will not be able to intersect $b$ and $e$ simultaneously. If $x$ lies to the right of $e$, then to intersect   $c$,   the height of $x$ must be smaller than $c$. Since $c$ does not intersect $a$,  $x$ will fail to intersect $a$.
  \begin{figure}[pt]
    \centering
    \includegraphics[width=.6\linewidth]{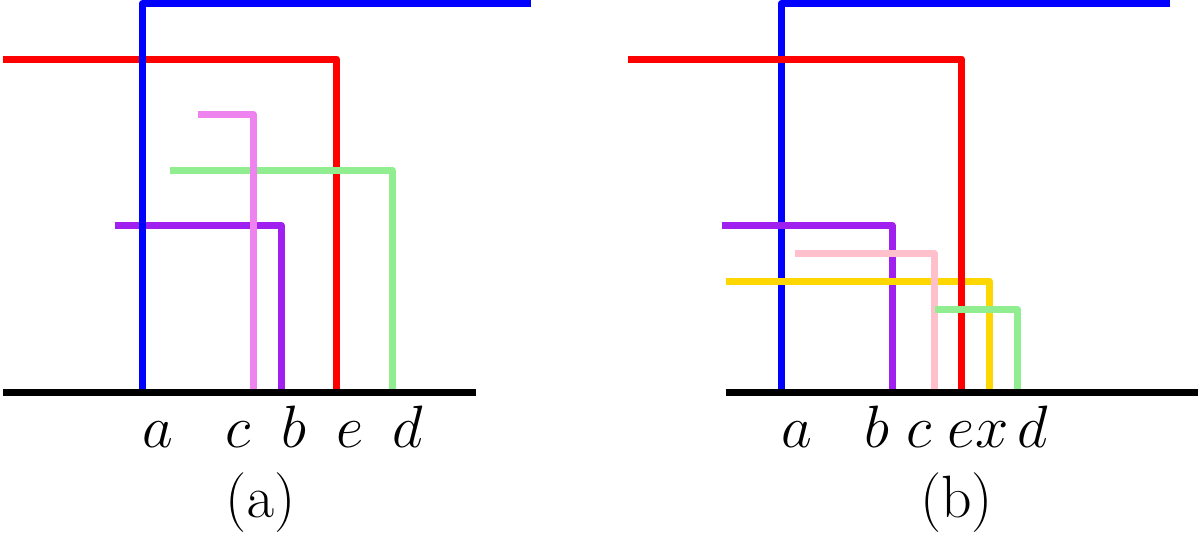}
    \caption{Illustration for (a) Case 2.1, and (b) Case 2.2. }
    \label{fig:case21}
\end{figure}
\end{proof}

\section{Conclusion}

In this paper we have examined the maximum clique problem for the grounded 1-bend string graphs. We show the problem to be hard for $y$-monotone strings. We also show that the problem  remains  hard  when we relax monotonicity, but restrict the bends and free endpoints of the  strings to  lie on three horizontal  lines. 

The most intriguing open problem is to settle the time complexity for grounded segment graphs. However, there are various questions worth investigating for 1-bend strings where the bends and endpoints lie on a few lines. If we allow only one horizontal line, then the resulting strings become segments and the corresponding intersection graph is a permutation graph, where one can find a maximum clique in polynomial time. For a fixed number of lines the problem is polynomial-time solvable for segments~\cite{MiddendorfP92}. Therefore, it would be interesting to examine whether the problem becomes polynomial-time solvable for two horizontal lines, where we can explore 1-bend strings. We think such restriction on the number of lines would be non-trivial even for strictly $y$-monotone  1-bend strings.

We have developed polynomial-time algorithms to find maximum clique for various types of $L$-shapes. A natural question is whether the running times of these algorithms can be improved. It would also be interesting to find non-trivial lower bounds on the time complexity.

\section*{Acknowledgements}

The research of Debajyoti Mondal is supported in part by the Natural Sciences and Engineering Research Council of Canada (NSERC).

\clearpage

\bibliographystyle{abbrvurl}
\bibliography{example-biblio}

\begin{thebibliography}{10}

\bibitem{mondal2019}
P.~Bose, P.~Carmi, M.~J. Keil, A.~Maheshwari, S.~Mehrabi, D.~Mondal, and
  M.~Smid.
\newblock Computing maximum independent set on outerstring graphs and their
  relatives.
\newblock In {\em Workshop on Algorithms and Data Structures}, pages 211--224.
  Springer, 2019.
\newblock \href {https://doi.org/10.1007/978-3-030-24766-9\_16}
  {\path{doi:10.1007/978-3-030-24766-9\_16}}.

\bibitem{cabello2012clique}
S.~Cabello, J.~Cardinal, and S.~Langerman.
\newblock The clique problem in ray intersection graphs.
\newblock In {\em European Symposium on Algorithms}, pages 241--252. Springer,
  2012.
\newblock \href {https://doi.org/10.1007/s00454-013-9538-5}
  {\path{doi:10.1007/s00454-013-9538-5}}.

\bibitem{DBLP:journals/jgaa/CardinalFMTV18}
J.~Cardinal, S.~Felsner, T.~Miltzow, C.~Tompkins, and B.~Vogtenhuber.
\newblock Intersection graphs of rays and grounded segments.
\newblock {\em J. Graph Algorithms Appl.}, 22(2):273--295, 2018.
\newblock \href {https://doi.org/10.7155/jgaa.00470}
  {\path{doi:10.7155/jgaa.00470}}.

\bibitem{ChlebikC07}
M.~Chleb{\'{\i}}k and J.~Chleb{\'{\i}}kov{\'{a}}.
\newblock The complexity of combinatorial optimization problems on
  $d$-dimensional boxes.
\newblock {\em {SIAM} J. Discret. Math.}, 21(1):158--169, 2007.
\newblock \href {https://doi.org/10.1137/050629276}
  {\path{doi:10.1137/050629276}}.

\bibitem{thesis}
P.~Chmel.
\newblock Algorithmic aspects of intersection representations, June 2020.
\newblock Bachelor Thesis, Charles University.

\bibitem{FoxP11}
J.~Fox and J.~Pach.
\newblock Computing the independence number of intersection graphs.
\newblock In D.~Randall, editor, {\em Proceedings of the Twenty-Second Annual
  {ACM-SIAM} Symposium on Discrete Algorithms (SODA)}, pages 1161--1165.
  {SIAM}, 2011.
\newblock \href {https://doi.org/10.1137/1.9781611973082.87}
  {\path{doi:10.1137/1.9781611973082.87}}.

\bibitem{Jelinek019}
V.~Jel{\'{\i}}nek and M.~T{\"{o}}pfer.
\newblock On grounded l-graphs and their relatives.
\newblock {\em Electr. J. Comb.}, 26(3):P3.17, 2019.
\newblock URL:
  \url{https://www.combinatorics.org/ojs/index.php/eljc/article/view/v26i3p17}.

\bibitem{keil2017algorithm}
J.~M. Keil, J.~S. Mitchell, D.~Pradhan, and M.~Vatshelle.
\newblock An algorithm for the maximum weight independent set problem on
  outerstring graphs.
\newblock {\em Computational Geometry}, 60:19--25, 2017.
\newblock \href {https://doi.org/10.1016/j.comgeo.2016.05.001}
  {\path{doi:10.1016/j.comgeo.2016.05.001}}.

\bibitem{kratochvil1990independent}
J.~Kratochv{\'\i}l and J.~Ne{\v{s}}et{\v{r}}il.
\newblock Independent set and clique problems in intersection-defined classes
  of graphs.
\newblock {\em Commentationes Mathematicae Universitatis Carolinae},
  31(1):85--93, 1990.

\bibitem{matousek2014intersection}
J.~Matousek.
\newblock Intersection graphs of segments and $\exists\mathbb {R} $.
\newblock {\em arXiv preprint arXiv:1406.2636}, 2014.

\bibitem{MiddendorfP92}
M.~Middendorf and F.~Pfeiffer.
\newblock The max clique problem in classes of string-graphs.
\newblock {\em Discret. Math.}, 108(1-3):365--372, 1992.
\newblock \href {https://doi.org/10.1016/0012-365X(92)90688-C}
  {\path{doi:10.1016/0012-365X(92)90688-C}}.

\bibitem{Nash10}
N.~Nash and D.~Gregg.
\newblock An output sensitive algorithm for computing a maximum independent set
  of a circle graph.
\newblock {\em Inf. Process. Lett.}, 110(16):630--634, 2010.
\newblock \href {https://doi.org/10.1016/j.ipl.2010.05.016}
  {\path{doi:10.1016/j.ipl.2010.05.016}}.

\bibitem{DBLP:conf/compgeom/Nekrich20}
Y.~Nekrich.
\newblock Four-dimensional dominance range reporting in linear space.
\newblock In {\em 36th International Symposium on Computational Geometry, SoCG
  2020, June 23-26, 2020, Z{\"{u}}rich, Switzerland}, pages 59:1--59:14, 2020.
\newblock \href {https://doi.org/10.4230/LIPIcs.SoCG.2020.59}
  {\path{doi:10.4230/LIPIcs.SoCG.2020.59}}.

\bibitem{pergel2016edge}
M.~Pergel and P.~Rzewski.
\newblock On edge intersection graphs of paths with 2 bends.
\newblock In {\em International Workshop on Graph-Theoretic Concepts in
  Computer Science}, pages 207--219. Springer, 2016.
\newblock \href {https://doi.org/10.1016/j.dam.2017.04.023}
  {\path{doi:10.1016/j.dam.2017.04.023}}.

\bibitem{tiskin2015fast}
A.~Tiskin.
\newblock Fast distance multiplication of unit-monge matrices.
\newblock {\em Algorithmica}, 71(4):859--888, 2015.
\newblock \href {https://doi.org/10.1007/s00453-013-9830-z}
  {\path{doi:10.1007/s00453-013-9830-z}}.

\bibitem{YuTC93}
M.~Yu, L.~Tseng, and S.~Chang.
\newblock Sequential and parallel algorithms for the maximum-weight independent
  set problem on permutation graphs.
\newblock {\em Inf. Process. Lett.}, 46(1):7--11, 1993.
\newblock \href {https://doi.org/10.1016/0020-0190(93)90188-F}
  {\path{doi:10.1016/0020-0190(93)90188-F}}.

\end{thebibliography}

\end{document}